\documentclass[12pt]{article}

\usepackage{amsmath}
\usepackage{amssymb}
\usepackage{amsfonts}
\usepackage{latexsym}
\usepackage{color}

\catcode `\@=11 \@addtoreset{equation}{section}

\catcode `\@=12

\newcommand{\be}{\begin{equation}}
\newcommand{\en}{\end{equation}}
\newcommand{\bea}{\begin{eqnarray}}
\newcommand{\ena}{\end{eqnarray}}
\newcommand{\beano}{\begin{eqnarray*}}
\newcommand{\enano}{\end{eqnarray*}}
\newcommand{\bee}{\begin{enumerate}}
\newcommand{\ene}{\end{enumerate}}

\newcommand{\mc}{\mathcal}

\newcommand{\D}{{\mc D}}

\newcommand{\Sc}{{\cal S}}
\newcommand{\G}{{\cal G}}

\newcommand{\F}{{\cal F}}

\newcommand{\Lc}{{\cal L}}

\newcommand{\1}{1 \!\! 1}
\newcommand{\N}{\mc N}

\newcommand{\Hil}{\mc H}

\newtheorem{thm}{Theorem}

\newtheorem{prop}[thm]{Proposition}
\newtheorem{defn}[thm]{Definition}

\newenvironment{proof}{\noindent {\bf Proof --}}{\hfill$\square$ \vspace{3mm}\endtrivlist}

\catcode `\@=11 \@addtoreset{equation}{section}
\catcode `\@=12

\begin{document}

\thispagestyle{empty}

\vspace*{2cm}

\begin{center}
{\Large \bf Deformed quons and bi-coherent states}   \vspace{2cm}\\

{\large F. Bagarello}\\

\vspace{3mm}
 DEIM -Dipartimento di Energia, ingegneria dell' Informazione e modelli Matematici,
\\ Scuola Politecnica, Universit\`a di Palermo, I-90128  Palermo, Italy\\
and INFN, Sezione di Napoli\\
e-mail: fabio.bagarello@unipa.it

\end{center}

\vspace*{2cm}

\begin{abstract}
\noindent

We discuss how a q-mutation relation can be deformed replacing a pair of conjugate operators with two other and unrelated operators,  as it is done in the construction of pseudo-fermions, pseudo-bosons and truncated pseudo-bosons. This deformation involves interesting mathematical problems and suggests possible applications to pseudo-hermitian quantum mechanics. We construct bi-coherent  states  associated to $\D$-pseudo-quons, and we show that they share  many of their properties with ordinary coherent states. In particular, we find conditions for these states to exist, to be eigenstates of suitable annihilation operators and to give rise to a resolution of the identity.
Two examples are discussed in details, one connected to an unbounded similarity map, and the other to a bounded map.

\end{abstract}

\vspace{2cm}

%{\bf PACS Numbers}:  .......

\vfill

%\pagenumbering{roman}

\newpage

\section{Introduction and preliminaries}\label{sectI}

The role of commutators in quantum mechanics is essential, both because of their technical utility, for instance in the deduction of the equations of notion from a given Hamiltonian, and because they are related to the statistical properties of the elementary particles they are {\em attached } to. Along the years it has become more and more evident that the canonical commutation or anticommutation relations (CCR and CAR) are not the only ones which play a relevant role in concrete applications. For instance, interesting alternatives are provided by the truncated version of CCR considered in \cite{truncCCR}, or by the commutation rules arising in the context of the so-called generalized Heisenberg algebra, see \cite{curado} and references therein.

Another commutation rule which has been proposed several years ago is, in our opinion, particularly interesting since it intertwines between CCR and CAR. This deformation produces the so-called quons, \cite{moh,fiv,gre}. In this case one assumes the following q-mutation relation between two ladder operators $c$ and $c^\dagger$, one the adjoint of the other: $ [c,c^\dagger]_q:=c c^\dagger -q c^\dagger c
=\1$, where $q\in [-1,1]$. Of course, when $q=1$ we recover CCR, while CAR are obtained for $q=-1$.

In all these relations a crucial aspect is that the (anti)-commutators are defined between two (or more) operators which are directly connected by the adjoint map, and the consequence is that one such operator behaves as a lowering, and its adjoint as a raising operator, and the product of the two, with the raising operator on the left, is a self-adjoint number (or number-like) operator. Recently, mainly in connection with the so-called pseudo-hermitian quantum mechanics and its relatives, \cite{bagbook,mosta,ben,specissue}, it has become clear that a physical observable does not need to be self-adjoint, or even symmetric, to have a concrete meaning. In particular, non self-adjoint Hamiltonians have been considered by several authors which give rise to interesting quantum dynamical systems. Among other features, these operators have only real eigenvalues. In many of such examples a number-like operator can be introduced, which is manifestly not self-adjoint since it can be factorized as follows: $N=BA$, with $B\neq A^\dagger$. Of course, this implies that $N^\dagger\neq N$, and since $N^\dagger$ turns out to be also a number-like operator, the situation is, in a sense, richer than  when assuming that the observables are all self-adjoint. We refer to \cite{baginbagbook} for many details and examples of this situation.

In this paper we show that a similar possibility can also be considered for quons, giving rise to an extended version of the q-mutator above in which the two operators involved are not one the adjoint of the other. This has several mathematical consequences, as we will see later. We will also discuss how pairs of coherent states can be constructed and how these are connected to the lowering operators which naturally appear in our construction.

The paper is organized as follows: in the rest of this section we briefly review few standard facts on quons and on pseudo-bosons, in order to keep the paper self-contained. In Section \ref{sectDPQs} we introduce the main object of our analysis, the $\D$-pseudo quons ($\D$-PQs), while in Section \ref{sectBCs} we define the so-called bi-coherent states (BCs) associated to our $\D$-PQs, and we deduce some of their properties. Section \ref{sectex} contains two detailed examples, while our conclusions are given in Section \ref{sectconcl}.

\subsection{Basic facts on quons}\label{subsetcI1}

We briefly review here few basic facts on quons, \cite{moh,fiv,gre}, useful for our generalization.

Quons are
defined essentially by their q-mutation relation \be [c,c^\dagger]_q:=c c^\dagger -q c^\dagger c
=\1, \qquad q\in [-1,1], \label{21} \en
between the creation and the annihilation operators $c^\dagger$
and $c$, which reduces to the CCR for $q=1$ and to the CAR for $q=-1$. For  $q$ in
the interval $]-1,1[$, equation (\ref{21}) describes particles which are neither
bosons nor fermions.  For completeness we also remark that other possible q-mutation relations have also been proposed along the years, but they will not be considered here.

In \cite{moh} it is proved that the eigenstates of $N_0=c^\dagger\,c$ are analogous to the bosonic ones, except that for the normalization. A simple concrete realization of the (\ref{21}) can be deduced as follows: let $\F_e=\{e_k,\, k=0,1,2,\ldots\}$ be the canonical orthonormal (o.n.) basis in $\Hil=l^2(\Bbb N_0)$, with all zero entries except in the $(k+1)$-th position, which is equal to one: $\left<e_k,e_m\right>=\delta_{k,m}$. If we take
\be
c=\left(
    \begin{array}{ccccccc}
      0 & \beta_0 & 0 & 0 & 0 & 0 & \cdots \\
      0 & 0 & \beta_1 & 0 & 0 & 0 & \cdots \\
      0 & 0 & 0 & \beta_2 & 0 & 0 & \cdots \\
      0 & 0 & 0 & 0 & \beta_3 & 0 & \cdots \\
      0 & 0 & 0 & 0 & 0 & \beta_4  & \cdots \\
      \cdots & \cdots & \cdots & \cdots & \cdots & \cdots & \cdots \\
      \cdots & \cdots & \cdots & \cdots & \cdots & \cdots & \cdots \\
    \end{array}
  \right),
\label{21bis}\en
it follows that (\ref{21}) is satisfied if $\beta_0^2=1$ and $\beta_n^2=1+q\beta_{n-1}^2$, $n\geq1$. Then $\beta_n^2$ coincides with $n+1$ if $q=1$, and with $\frac{1-q^{n+1}}{1-q}$ if $q\neq1$. Here we will always consider $\beta_n>0$ for all $n\geq0$. Moreover $c\,e_0=0$, and $c^\dagger$ behaves as a raising operator since from (\ref{21bis}) we deduce
\be
e_{n+1}=\frac{1}{\beta_n}\,c^\dagger e_n=\frac{1}{\beta_n!}(c^\dagger)^{n+1}\,e_0,
\label{22}\en
for all $n\geq0$. Here we have introduced $\beta_n!:=\beta_n\beta_{n-1}\cdots\beta_2\beta_1$. In the literature this quantity is sometimes called the {\em $q$-factorial}. Moreover, rather than $\beta_n$, sometimes the symbol $[n]$ (or even $[n]_q)$  is used, \cite{eremel}, with $\beta_n^2=[n+1]$ . Of course, from (\ref{22}) it follows that  $c^\dagger e_n=\beta_ne_{n+1}$. Using (\ref{21bis}) it is also easy to check that $c$ acts as a lowering operator on $\F_e$: $c\,e_m=\beta_{m-1}e_{m-1}$, for all $m\in\Bbb N_0$, where we have also introduced $\beta_{-1}=0$, to ensure that $c\,e_0=0$.

 Then we have
\be
N_0e_m=\beta_{m-1}^2e_m,
\label{23}\en
for all $m\in \Bbb N_0$. The operator $N$, formally defined in \cite{moh} as follows $N=\frac{1}{\log(q)}\,\log(\1-N_0(1-q))$ for $q>0$, satisfies the eigenvalue equation $Ne_m=me_m$, for all $m\in \Bbb N_0$.

It should be stressed that the one in (\ref{21bis}) is not the only possible way to represent the operator $c$. For instance, in \cite{eremel}, the authors adopt the following representation of $c$ and $c^\dagger$ in $\Lc^2(\Bbb R)$:
\be
c=\frac{e^{-2i\alpha x}-e^{i\alpha \frac{d}{dx}}e^{-i\alpha x}}{-i\sqrt{1-e^{-2\alpha^2}}}, \qquad\qquad c^\dagger=\frac{e^{2i\alpha x}-e^{i\alpha x}e^{i\alpha \frac{d}{dx}}}{i\sqrt{1-e^{-2\alpha^2}}},
\label{23bis}\en
where $\alpha=\sqrt{-\frac{\log(q)}{2}}$ or, which is the same, $q=e^{-2\alpha^2}$. However, since $\alpha$ is assumed to belong to the set  $[0,\infty)$,  $q$ ranges in the interval $]0,1]$. Then, the representation in (\ref{23bis}) only works in this region. We will go back to this representation in Section \ref{subsectex1}, while the matrix form in (\ref{21bis}) will be used in Section \ref{subsectex2}.

\subsection{Basic facts on $\D$-PBs}\label{subsetcI2}

Another fundamental ingredient of this paper is what arises out of a different deformation of the CCR, the one giving rise to the so-called $\D$-pseudo bosons ($\D$-PBs), \cite{baginbagbook}. Here we briefly collect some of the definitions and results on $\D$-PBs which will be useful in the following.

Let $\Hil$ be a given Hilbert space with scalar product $\left<.,.\right>$ and related norm $\|.\|$. Let further $A$ and $B$ be two operators
on $\Hil$, with domains $D(A)$ and $D(B)$ respectively, $A^\dagger$ and $B^\dagger$ their adjoint, and let $\D$ be a dense subspace of $\Hil$
such that $A^\sharp\D\subseteq\D$ and $B^\sharp\D\subseteq\D$, where $X^\sharp$ is either $X$ or $X^\dagger$. Of course, $\D\subseteq D(A^\sharp)$
and $\D\subseteq D(B^\sharp)$.

\begin{defn}\label{def21}
The operators $(A,B)$ are $\D$-pseudo bosonic  if, for all $f\in\D$, we have
\be
A\,B\,f-B\,A\,f=f.
\label{A1}\en
\end{defn}

\vspace{2mm}

Our  working assumptions are the following:

\vspace{2mm}

{\bf Assumption $\D$-pb 1.--}  there exists a non-zero $\hat\varphi_{ 0}\in\D$ such that $A\,\hat\varphi_{ 0}=0$.

\vspace{1mm}

{\bf Assumption $\D$-pb 2.--}  there exists a non-zero $\hat\Psi_{ 0}\in\D$ such that $B^\dagger\,\hat\Psi_{ 0}=0$.

\vspace{2mm}

Then, if $(A,B)$ satisfy Definition \ref{def21}, it is obvious that $\hat\varphi_0\in D^\infty(B):=\cap_{k\geq0}D(B^k)$ and that $\hat\Psi_0\in D^\infty(A^\dagger)$, so
that the vectors \be \hat\varphi_n:=\frac{1}{\sqrt{n!}}\,B^n\hat\varphi_0,\qquad \hat\Psi_n:=\frac{1}{\sqrt{n!}}\,{A^\dagger}^n\hat\Psi_0, \label{A2a}\en
$n\geq0$, can be defined and they all belong to $\D$. As a consequence, they belong to the domains of $A^\sharp$, $B^\sharp$ and $\N^\sharp$, where $\N=BA$. We call $\F_{\hat\Psi}=\{\hat\Psi_{ n}, \,n\geq0\}$ and
$\F_{\hat\varphi}=\{\hat\varphi_{ n}, \,n\geq0\}$.

It is now simple to deduce the following lowering and raising relations:
\be
\left\{
    \begin{array}{ll}
B\,\hat\varphi_n=\sqrt{n+1}\hat\varphi_{n+1}, \qquad\qquad\quad\,\, n\geq 0,\\
A\,\hat\varphi_0=0,\quad A\hat\varphi_n=\sqrt{n}\,\hat\varphi_{n-1}, \qquad\,\, n\geq 1,\\
A^\dagger\hat\Psi_n=\sqrt{n+1}\hat\Psi_{n+1}, \qquad\qquad\quad\, n\geq 0,\\
B^\dagger\hat\Psi_0=0,\quad B^\dagger\hat\Psi_n=\sqrt{n}\,\hat\Psi_{n-1}, \qquad n\geq 1,\\
       \end{array}
        \right.
\label{A3}\en as well as the eigenvalue equations $\N\hat\varphi_n=n\hat\varphi_n$ and  $\N^\dagger\hat\Psi_n=n\hat\Psi_n$, $n\geq0$. In particular, as a consequence
of these  last two equations,  choosing the normalization of $\hat\varphi_0$ and $\hat\Psi_0$ in such a way $\left<\hat\varphi_0,\hat\Psi_0\right>=1$ is satisfied, we deduce that
\be \left<\hat\varphi_n,\hat\Psi_m\right>=\delta_{n,m}, \label{A4}\en
 for all $n, m\geq0$. Hence $\F_{\hat\Psi}$ and $\F_{\hat\varphi}$ are biorthogonal. Our third assumption is the following:

\vspace{2mm}

{\bf Assumption $\D$-pb 3.--}  $\F_{\hat\varphi}$ is a basis for $\Hil$.

\vspace{1mm}

This is equivalent to requiring that $\F_{\hat\Psi}$ is a basis for $\Hil$ as well, \cite{chri}. However, several  physical models suggest to adopt the following weaker version of this assumption, \cite{baginbagbook}:

\vspace{2mm}

{\bf Assumption $\D$-pbw 3.--}  For some subspace $\G$ dense in $\Hil$, $\F_{\hat\varphi}$ and $\F_{\hat\Psi}$ are $\G$-quasi bases.

\vspace{2mm}
This means that, for all $f$ and $g$ in $\G$,
\be
\left<f,g\right>=\sum_{n\geq0}\left<f,\hat\varphi_n\right>\left<\hat\Psi_n,g\right>=\sum_{n\geq0}\left<f,\hat\Psi_n\right>\left<\hat\varphi_n,g\right>,
\label{A4b}
\en
which can be seen as a weak form of the resolution of the identity, restricted to $\D$. The role, and the necessity, of these sets is discussed in details in \cite{baginbagbook}, and in some more recent papers, \cite{bagadd1,bagadd2,bagadd3}.

As suggested by concrete physical systems, in \cite{baginbagbook} we have also assumed that a self-adjoint, invertible, operator $\Theta$, which leaves, together with $\Theta^{-1}$, $\D$ invariant, exists: $\Theta\D\subseteq\D$, $\Theta^{-1}\D\subseteq\D$. Then we say that $(A,B^\dagger)$ are $\Theta-$conjugate if $Af=\Theta^{-1}B^\dagger\,\Theta\,f$, for all $f\in\D$. One can prove that, if $\F_{\hat\varphi}$ and $\F_{\hat\Psi}$ are $\D$-quasi bases for $\Hil$, then the operators $(A,B^\dagger)$ are $\Theta-$conjugate if and only if $\hat\Psi_n=\Theta\hat\varphi_n$, for all $n\geq0$. Moreover, if $(A,B^\dagger)$ are $\Theta-$conjugate, then $\left<f,\Theta f\right>>0$ for all non zero $f\in \D$.

We refer to \cite{baginbagbook} for more results on $\Theta$, which sometimes is also used to define different scalar products in $\Hil$, on  $\D$-PBs, and for some explicit examples of their physical appearances in the quantum mechanical literature.

\section{$\D$-pseudo quons}\label{sectDPQs}

The aim of this section is to show that the two deformations considered in Section \ref{sectI} can be somehow {\em merged}, giving rise to what we call $\D$-pseudo quons, ($\D$-PQs). The procedure is very close to that for ordinary bosons, but requires some cares mainly for normalization problems, but not only.

Let $\Hil$ be a given Hilbert space with scalar product $\left<.,.\right>$ and related norm $\|.\|$. Let further $a$ and $b$ be two operators
on $\Hil$, with domains $D(a)$ and $D(b)$ respectively, $a^\dagger$ and $b^\dagger$ their adjoint, and let $\D$ be a dense subspace of $\Hil$
such that $a^\sharp\D\subseteq\D$ and $b^\sharp\D\subseteq\D$. As in Section \ref{subsetcI2}, this implies, in particular, that $\D\subseteq D(a^\sharp)$
and $\D\subseteq D(b^\sharp)$.

\begin{defn}\label{def31}
The operators $(a,b)$ are $\D$-pseudo quonic  if, for all $f\in\D$, we have
\be
[a,b]_qf:=a\,b\,f-q \,b\,a\,f=f,
\label{31}\en
for some real $q$.
\end{defn}

\vspace{2mm}

Notice that, for the time being, we are not interested in restricting the values of $q$ to the closed interval $[-1,1]$. Our initial working assumptions coincide with those for $\D$-PBs:

\vspace{2mm}

{\bf Assumption $\D$-pq 1.--}  there exists a non-zero $\varphi_{ 0}\in\D$ such that $a\,\varphi_{ 0}=0$.

\vspace{1mm}

{\bf Assumption $\D$-pq 2.--}  there exists a non-zero $\Psi_{ 0}\in\D$ such that $b^\dagger\,\Psi_{ 0}=0$.

\vspace{2mm}

Then, if $(a,b)$ satisfy Definition \ref{def31}, it is obvious that the vectors \be \varphi_n:=\frac{1}{{\beta_{n-1}!}}\,b^n\varphi_0=\frac{1}{\beta_{n-1}}\,b\,\varphi_{n-1},\qquad \Psi_n:=\frac{1}{{\beta_{n-1}!}}\,{a^\dagger}^n\Psi_0=\frac{1}{{\beta_{n-1}}}\,{a^\dagger}\,\Psi_{n-1}, \label{A2}\en
$n\geq1$, can be defined and they all belong to $\D$ and, as a consequence, to the domains of $a^\sharp$, $b^\sharp$ and $N^\sharp$, where $N=ba$. Here $\beta_n$ are those introduced in Section \ref{subsetcI1}, for ordinary quons. Remember that each $\beta_n$ is strictly positive. We further define $\F_\Psi=\{\Psi_{ n}, \,n\geq0\}$ and
$\F_\varphi=\{\varphi_{ n}, \,n\geq0\}$.

Then the following lowering and raising relations arise, which replace those in (\ref{A3}):
\be
\left\{
    \begin{array}{ll}
b\,\varphi_n=\beta_n\varphi_{n+1}, \qquad\qquad\qquad\qquad\qquad\quad\,\, n\geq 0,\\
a\,\varphi_0=0,\quad a\varphi_n=\beta_{n-1}\,\varphi_{n-1}, \qquad\qquad\quad\,\, n\geq 1,\\
a^\dagger\Psi_n=\beta_n\Psi_{n+1}, \qquad\qquad\qquad\qquad\quad\quad\quad n\geq 0,\\
b^\dagger\Psi_0=0,\quad b^\dagger\Psi_n=\beta_{n-1}\,\Psi_{n-1}, \qquad\quad\qquad n\geq 1.\\
       \end{array}
        \right.
\label{32}\en
In a more compact way we can rewrite the second and the fourth of these equalities as follows: $a\varphi_n=\beta_{n-1}\,\varphi_{n-1}$ and $b^\dagger\Psi_n=\beta_{n-1}\,\Psi_{n-1}$, with the agreement that $\Psi_{-1}=\varphi_{-1}$ are the zero vectors. Recall also that $\beta_{-1}=0$. Then, from (\ref{32}), using the stability of the set $\D$ under the action of $b$ and $a^\dagger$, we deduce the eigenvalue equations
$N\varphi_n=\beta_{n-1}\varphi_n$ and  $N^\dagger\Psi_n=\beta_{n-1}\Psi_n$, $n\geq0$ which show that $N$ and $N^\dagger$ are isospectrals. This suggests the existence of some intertwining operator between $N$ and $N^\dagger$. This will be discussed later. Furthermore,  choosing the normalization of $\varphi_0$ and $\Psi_0$ in such a way $\left<\varphi_0,\Psi_0\right>=1$, we conclude, as in Section \ref{subsetcI2}, that
\be \left<\varphi_n,\Psi_m\right>=\delta_{n,m}, \label{33}\en
 for all $n, m\geq0$. Hence $\F_\Psi$ and $\F_\varphi$ are biorthogonal. Our third assumption is the following:

\vspace{2mm}

{\bf Assumption $\D$-pq 3.--}  $\F_\varphi$ is a basis for $\Hil$,

\vspace{1mm}

(which implies that $\F_\Psi$ is a basis for $\Hil$ as well, \cite{chri}), or its weaker version:

\vspace{2mm}

{\bf Assumption $\D$-pqw 3.--}  For some subspace $\G$ dense in $\Hil$, $\F_\varphi$ and $\F_\Psi$ are $\G$-quasi bases.

\vspace{2mm}
As in \cite{baginbagbook} we assume here that a self-adjoint, invertible, operator $\Theta$, which leaves, together with $\Theta^{-1}$, $\D$ invariant, exists: $\Theta\D\subseteq\D$, $\Theta^{-1}\D\subseteq\D$. Then  $(a,b^\dagger)$ are $\Theta-$conjugate if $af=\Theta^{-1}b^\dagger\,\Theta\,f$, for all $f\in\D$. Despite of the differences between (\ref{21}) and (\ref{31}),  one can prove again that, if $\F_\varphi$ and $\F_\Psi$ are $\D$-quasi bases for $\Hil$, then the operators $(a,b^\dagger)$ are $\Theta-$conjugate if and only if $\Psi_n=\Theta\varphi_n$, for all $n\geq0$. Moreover, under the same assumption, we easily see that $N^\dagger\Theta\varphi_n=\Theta N\varphi_n$ for all $n$, or, more in general, that $\left(N^\dagger\Theta-\Theta N\right)f=0$ for all $f\in \Lc_\varphi$, the linear span of the $\varphi_n$'s, which is dense in $\Hil$ under either Assumption $\D$-pq 3 or its $\D$-pqw 3 version. This is a weak form of the intertwining relation $N^\dagger\Theta=\Theta N$, which could be further established in all of $\Hil$ if $\Theta$ and $N$ are both bounded. Hence we conclude that $\D$-PBs and $\D$-PQs are not really different from many points of view.

\vspace{2mm}

{\bf Remark:--} If $q\in [-1,1[$ then the operator $c$ in Section \ref{subsetcI1} turns out to be bounded, while it is not so if $q\geq1$, \cite{bagadd4}. This suggests that the situation for $a$ and $b$ here is even more complicated. In fact, a simple estimate shows that
$$
\|a\|^2=\sup_{\|f\|=1}\|af\|^2\geq \beta^2_{n-1}\left(\frac{\|\varphi_{n-1}\|}{\|\varphi_n\|}\right)^2,
$$
and the right-hand side can be convergent or divergent with $n$ depending on the value of $q$ and of the ratio $\frac{\|\varphi_{n-1}\|}{\|\varphi_n\|}$. Whenever $\|a\|$ and $\|b\|$ are finite, we could take $\D$ coinciding with $\Hil$. But this choice in general is not enough to ensure that the Assumptions $\D$-pq1, $\D$-pq2 and $\D$-pq3 are valid, and they must be checked in explicit situations.

\vspace{2mm}

In view of this Remark, also $\Theta$ and $\Theta^{-1}$ could be unbounded, in principle. This is what happens, for instance, when $\F_\varphi$ and $\F_\Psi$ are not Riesz bases, \cite{baginbagbook}. In this case an explicit expression for these operator is
$$
D(\Theta)=\{f\in\Hil: \,\sum_n\left<\Psi_n,f\right>\Psi_n\in\Hil\}, \qquad D(\Theta^{-1})=\{h\in\Hil: \,\sum_n\left<\varphi_n,h\right>\varphi_n\in\Hil\},
$$
and
\be
\Theta f=\sum_n\left<\Psi_n,f\right>\Psi_n, \qquad \Theta^{-1}h=\sum_n\left<\varphi_n,h\right>\varphi_n,
\label{38}\en
for all $f\in D(\Theta)$ and $h\in D(\Theta^{-1})$. We will discuss in a moment in which sense $\Theta^{-1}$ defined here is the inverse of $\Theta$. First we notice that each $\varphi_n\in D(\Theta)$, while each $\Psi_n\in D(\Theta^{-1})$, and that $\Theta\varphi_n=\Psi_n$, $\Theta^{-1}\Psi_n=\varphi_n$, for all $n$. Hence, if $\F_\varphi$ and $\F_\Psi$ are at least complete in $\Hil$, $\Theta$ and $\Theta^{-1}$ are densely defined. This is because, for instance, $D(\Theta)$ contains $\Lc_\varphi$, which is dense in $\Hil$.

Finally, if $(\F_\varphi,\F_\Psi)$ are $\D$-quasi bases, using the continuity of the scalar product we can prove that
\be
\left<\Theta f,\Theta^{-1}h\right>=\left<\sum_n\left<\Psi_n,f\right>\Psi_n,\Theta^{-1}h\right>=\sum_n\left<f,\Psi_n\right>
\left<\Psi_n,\Theta^{-1}h\right>=\sum_n\left< f,\Psi_n\right>\left<\varphi_n,h\right>=\left<f,h\right>,
\label{34}\en
$f,g\in\D$, if $\D\subseteq D(\Theta)\cap D(\Theta^{-1})$. Then we have $\left< f,\Theta \Theta^{-1}h\right>=\left<f,h\right>$, and, using the density of $\D$ in $\Hil$,  $\Theta \Theta^{-1}h=h$ for all $h\in\D$. In a similar way, we prove that $\Theta^{-1} \Theta h=h$. Hence, the operator $\Theta^{-1}$ defined as in (\ref{38}) is really the inverse of $\Theta$, as stated, at least in $\D$. Of course, if both $\Theta$ and $\Theta^{-1}$ are bounded,  equation (\ref{34}) can be extended to all of $\Hil$.

\section{Bicoherent states}\label{sectBCs}

During recent years several possible definitions of coherent states have been proposed for quons, \cite{eremel}, \cite{kar}-\cite{baz}. More recently, coherent states related to $\D$-PBs have also been proposed and studied, \cite{bagcs1}-\cite{bagcs3}. They have been called bi-coherent states (BCs) because they always {\em go in pairs}. This is not really surprising, since $\D$-PBs admit two different lowering operators, see $A$ and $B^\dagger$ in (\ref{A3}). In this section we propose a definition of BCs for our $\D$-PQs. The conclusion of our analysis will be that replacing the quons commutation rule with its deformed version will cause no particular problem, meaning with this that most of the properties of coherent states will be recovered (in an extended version).

To avoid useless complication we will restrict here to $0<q<1$.  The following Proposition, which adapts to quons the results discussed in \cite{bagcs2,bagcs3}, holds:

\begin{prop}\label{thm1} Let $(\F_\varphi,\F_\Psi)$ be $\D$-quasi bases for $\Hil$, and let  us assume that there exist four positive constants $A_\varphi, A_\Psi, r_\varphi, r_\Psi >0$, and two positive sequences $M_n(\varphi), M_n(\Psi)>0$, such that $\|\varphi_n\|\leq A_\varphi r_\varphi^n M_n(\varphi)$ and $\|\Psi_n\|\leq A_\Psi r_\Psi^n M_n(\Psi)$, for all $n\geq0$. Suppose further that  $\lim_{n,\infty}\frac{M_n(\varphi)}{M_{n+1}(\varphi)}=M(\varphi)$ and $\lim_{n,\infty}\frac{M_n(\Psi)}{M_{n+1}(\Psi)}=M(\Psi)$, with $0<M(\varphi), M(\Psi)<\infty$. Let us define
\be
N(|z|)=\left(\sum_{k=0}^\infty\frac{|z|^{2k}}{(\beta_{k-1}!)^2}\right)^{-1/2},
\label{41}\en
and
\be
\varphi(z)=N(|z|) \sum_{k=0}^\infty\frac{z^{k}}{\beta_{k-1}!}\varphi_k, \qquad \Psi(z)=N(|z|) \sum_{k=0}^\infty\frac{z^{k}}{\beta_{k-1}!}\Psi_k.
\label{42}\en
Let
$$
\rho_\varphi = \min\left\{\frac{1}{\sqrt{1-q}},\frac{M(\varphi)}{r_\varphi}\,\frac{1}{\sqrt{1-q}}\right\},\qquad \rho_\Psi = \min\left\{\frac{1}{\sqrt{1-q}},\frac{M(\Psi)}{r_\Psi}\,\frac{1}{\sqrt{1-q}}\right\},
$$
and $\rho:=\min\left\{\rho_\varphi, \rho_\Psi\right\}$. Then $\varphi(z)$ (resp. $\Psi(z)$) is well defined for $|z|<\rho_\varphi$ (resp. $|z|<\rho_\Psi$). Moreover, for $|z|<\rho$, $\left<\varphi(z),\Psi(z)\right>=1$, $a\varphi(z)=z\,\varphi(z)$ and $b^\dagger\Psi(z)=z\,\Psi(z)$. Also, if a measure $d\lambda(r)$ exists such that
$\int_0^\rho d\lambda(r) r^{2k}=\frac{\beta_{k-1}^2!}{2\pi}$, for all $k\geq0$, then, calling $d\nu(z,\overline{z})=d\lambda(r)\, d\theta$, we have
\be
\int_{C_\rho(0)} d\nu(z,\overline{z}) N(|z|)^{-2}\left<f,\varphi(z)\right>\left<\Psi(z),g\right>=\left<f,g\right>,
\label{43}\en
for all $f,g\in\D$,  where $C_\rho(0)$ is the circle centered in the origin of the complex plane and of radius $\rho$.

\end{prop}

\begin{proof}
From (\ref{42}) it is evident that  $\varphi(z)$ and $\Psi(z)$ are well defined only inside that region of $\Bbb C$ where, first of all, $N(|z|)$ is defined. A standard computation shows that the series for $N(|z|)$ converges if $|z|^2<\lim_{k,\infty}\beta_k^2=\frac{1}{1-q}$, i.e. if $|z|<\frac{1}{\sqrt{1-q}}$.

Now, since by assumption $\|\varphi_n\|\leq A_\varphi r_\varphi^n M_n(\varphi)$, we get
$$
\|\varphi(z)\|\leq \left|N(|z|)\right|A_\varphi\sum_{n=0}^\infty
\left(\frac{M_n(\varphi)}{\beta_{n-1}!}\right)|zr_\varphi|^n.
$$
This power series converges if
$$
|zr_\varphi|<\lim_{n,\infty}\frac{M_n(\varphi)}{\beta_{n-1}!}\,\frac{\beta_{n}!}{M_{n+1}(\varphi)}=
\lim_{n,\infty}\beta_n\frac{M_n(\varphi)}{M_{n+1}(\varphi)}=\lim_{n,\infty}\sqrt{\frac{1-q^{n+1}}{1-q}}\,
\frac{M_n(\varphi)}{M_{n+1}(\varphi)}=\frac{M(\varphi)}{\sqrt{1-q}},
$$
i.e. if $|z|<\frac{M(\varphi)}{r_\varphi\sqrt{1-q}}$. Notice that, depending on the value of $\frac{M(\varphi)}{r_\varphi}$, this quantity defines a  circle of convergence which is bigger or  smaller than the one in which $N(|z|)$ converges. Hence, to be sure that $\varphi(z)$ is well defined, we have to introduce $\rho_\varphi = \min\left\{\frac{1}{\sqrt{1-q}},\frac{M(\varphi)}{r_\varphi\sqrt{1-q}}\right\}$, and to consider only those $z\in\Bbb C$ with $|z|<\rho_\varphi$.

In the same way we can check that $\Psi(z)$ is well defined if $|z|<\rho_\Psi$.

The fact that $a\varphi(z)=z\,\varphi(z)$ and $b^\dagger\Psi(z)=z\,\Psi(z)$, when they are defined, follows from the lowering equations for $a$ and $b^\dagger$ in (\ref{32}): $a\varphi_n=\beta_{n-1}\,\varphi_{n-1}$ and $b^\dagger\Psi_n=\beta_{n-1}\,\Psi_{n-1}$. Definition (\ref{41}) for $N(|z|)$, and the biorthogonality of $\F_\varphi$ and $\F_\Psi$, easily imply that $\left<\varphi(z),\Psi(z)\right>=1$ for all $|z|<\rho$.

The proof of the resolution of the identity is a minor extension of the one given in \cite{bagcs2}.

\end{proof}

Even if the idea of the construction of the BCs is apparently close to that discussed in \cite{bagcs2}, it should be  stressed that, while in \cite{bagcs2} the ladder operators were defined ad-hoc out of two biorthogonal sets of vectors, here these operators are the fundamental building blocks of the entire construction: we are reversing the starting point, using the operators to construct the relevant vectors rather than going the other way around.

Examples of sequences $M_n(\varphi)$ and $M_n(\Psi)$ which could be used in the estimates above are the following: (i) any  constant sequence $M_n(\varphi)=M$, $M>0$. In this case $M(\varphi)=1$; (ii) $M_n(\varphi)=M^n$, for some $M>0$. In this case $M(\varphi)=\frac{1}{M}$; (iii), $M_n(\varphi)=\gamma_n!$ for any sequence $\{\gamma_n\neq0\}$ converging to some $\hat\gamma\in]0,\infty[$. In this case $M(\varphi)=\frac{1}{\hat\gamma}$. Notice that, in particular, this is what happens if we take $\gamma_n=\beta_n$, while it is not the case if we take $\gamma_n=n$.

As for the uncertainty relation, this is saturated in the following (generalized) version: if we put $Q=\frac{1}{\sqrt{2}}(b+a)$ and $P=\frac{i}{\sqrt{2}}(b-a)$, and we call $\left<T\right>:=\left<\Psi(z),T\varphi(z)\right>$, for all operator $T$ with $\varphi(z)\in D(T)$, we deduce that
$$
\Delta Q\,\Delta P=\frac{1}{2}\left(|z|^2(q-1)+1\right).
$$
Here $\Delta Q$ and $\Delta P$ are defined as in $(\Delta T)^2=\left<T^2\right>-\left<T\right>^2$, for $T=Q,P$. Notice that this equality reduces to the standard one, $\Delta Q\,\Delta P=\frac{1}{2}$, when $q=1$, i.e. when we go back to $\D$-PBs.

\section{Examples}\label{sectex}

In this section we consider two examples of $\D-$PQs and of their related BCs. In the first case we will extend the example originally discussed in \cite{eremel}, see (\ref{23bis}), while the second is a deformation of the quons defined in $\Hil=l^2(\Bbb N_0)$ as in (\ref{21bis}).

\subsection{A deformation of the model by Eremin and Meldianov}\label{subsectex1}

As we have seen in the Introduction, it is possible to represent the quonic operators $c$ and $c^\dagger$ in $\Lc^2(\Bbb R)$ in terms of the multiplication and of the derivative operators $x$ and $\frac{d}{dx}$ as follows:
\be
c=\frac{e^{-2i\alpha x}-e^{i\alpha \frac{d}{dx}}e^{-i\alpha x}}{-i\sqrt{1-e^{-2\alpha^2}}}, \quad c^\dagger=\frac{e^{2i\alpha x}-e^{i\alpha x}e^{i\alpha \frac{d}{dx}}}{i\sqrt{1-e^{-2\alpha^2}}}.
\label{51}\en
Here $\alpha=\sqrt{-\frac{\log(q)}{2}}$ or, which is the same, $q=e^{-2\alpha^2}$. We remind that $\alpha\in[0,\infty)$, so that $q\in]0,1]$. It is not hard to deform $c$ in order to obtain operators satisfying (\ref{31}). For that, let $\gamma\in\Bbb R$ be a fixed real parameter, and let us introduce the operators
\be
a=\frac{e^{-2i\alpha x}-e^{i\alpha \frac{d}{dx}}e^{-i\alpha (x+\gamma)}}{-i\sqrt{1-e^{-2\alpha^2}}}, \quad b=\frac{e^{2i\alpha x}-e^{i\alpha (x-\gamma)}e^{i\alpha \frac{d}{dx}}}{i\sqrt{1-e^{-2\alpha^2}}}.
\label{52}\en
Of course, $(a,b)$ collapse to $(c,c^\dagger)$ if $\gamma=0$, but, if $\gamma\neq0$, then $b\neq a^\dagger$. Incidentally, it might be useful to stress that even this simple dependence on $\gamma$  can have non trivial consequences, as we have seen several times for $\D$-PBs, \cite{baginbagbook}. In those cases, in fact, it often happens that the set of eigenvectors of the two conjugated number-like operators are no longer bases for the Hilbert space where the model ilives, while they are still complete. For this reason the example we are going to discuss can be interesting, as it will appear clear in the following.

First of all, we want to show that $(a,b)$ are $\D$-PQs in the sense of Definition \ref{def31}, with $\D$ the dense subset of $\Lc^2(\Bbb R)$ already defined as in \cite{baginbagbook}, Section 3.3.1.3:
\be\D=\{f(x)\in \Sc(\Bbb R):\,e^{kx}f(x)\in\Sc(\Bbb R), \, \forall k\in \Bbb
R\}.\label{53}\en
First we observe that $\D$ is stable under the action of $a$, $b$ and their adjoint. Then, a direct computation shows first that, for all $f(x)\in\D$, $([a,b]_qf)(x)=f(x)$. Hence $a$ and $b$ satisfy condition (\ref{31}). The vacua of $a$ and $b^\dagger$ can be easily deduced, extending what is done in \cite{eremel}:
\be
\varphi_0(x)=\frac{1}{\pi^{1/4}}e^{-\frac{x^2}{2}+x\left(\gamma+\frac{3}{2}\,i\alpha\right)}, \qquad \Psi_0(x)=\frac{1}{\pi^{1/4}}e^{-\frac{x^2}{2}+x\left(-\gamma+\frac{3}{2}\,i\alpha\right)},
\label{54}\en
and $a\,\varphi_0(x)=b^\dagger\Psi_0(x)=0$. Also, $\varphi_0(x),\Psi_0(x)\in\D$, and the normalization is chosen to ensure that $\left<\varphi_0,\Psi_0\right>=1$.
Now, since $\D$ is stable under the action of both $b$ and $a^\dagger$, the vectors $\varphi_n$ and $\Psi_n$ can be defined as in (\ref{A2}), and we have $\varphi_n=\frac{1}{{\beta_{n-1}!}}\,b^n\varphi_0$ and $ \Psi_n=\frac{1}{{\beta_{n-1}!}}\,{a^\dagger}^n\Psi_0$, for all
$n\geq1$. Notice that it is enough to compute the vectors $\varphi_n(x)$, since the functions $\Psi_n(x)$ can easily be deduced by these ones simply replacing $\gamma$ with $-\gamma$. This is a consequence of the analytic expressions of $\varphi_0(x)$ and $\Psi_0(x)$, and of the relation between $a^\dagger$ and $b$. A direct computation shows that, for instance,
$$
\varphi_1(x)=\frac{-i}{\sqrt{1-e^{-2\alpha^2}}}\,\varphi_0(x)\left(e^{2i\alpha x}-e^{-\alpha^2}\right),
$$
$$
\varphi_2(x)=\frac{1}{\beta_1}\left(\frac{-i}{\sqrt{1-e^{-2\alpha^2}}}\right)^2\,\varphi_0(x)\left(e^{4i\alpha x}-e^{2i\alpha x-\alpha^2}-e^{2i\alpha x-3\alpha^2}+e^{-2\alpha^2}\right),
$$
and so on. For generic $n$ the function $\varphi_n(x)$ can be written as

$$
\varphi_n(x)=\frac{1}{\beta_{n-1}!}\,\left(\frac{-i}{\sqrt{1-e^{-2\alpha^2}}}\right)^n\,\varphi_0(x)\sum_{k=0}^nc_k^{(n)}e^{2i\alpha k x},
$$
where $c_k^{(n)}$ are suitable coefficients which should be computed for each $n$. For instance:
$$
c_0^{(0)}=1, \quad c_0^{(1)}=-e^{-\alpha^2} \mbox{  and } c_1^{(1)}=1, \quad c_0^{(2)}=-e^{-2\alpha^2},\, c_1^{(2)}=-e^{-\alpha^2}-e^{-3\alpha^2}, \mbox{  and } c_2^{(2)}=1,
$$
and so on. The same coefficients appear in the expression for $\Psi_n(x)$ which looks like
$$
\Psi_n(x)=\frac{1}{\beta_{n-1}!}\,\left(\frac{-i}{\sqrt{1-e^{-2\alpha^2}}}\right)^n\,\Psi_0(x)\sum_{k=0}^nc_k^{(n)}e^{2i\alpha k x}.
$$
This is due to the fact that $c_k^{(n)}$ does not depend on $\gamma$. We will return on the relation between  these functions and those in \cite{eremel} later on. Now, we check that the two sets $\F_\varphi=\{\varphi_n(x),\,n\geq0\}$ and $\F_\Psi=\{\Psi_n(x),\,n\geq0\}$ are both complete in $\Lc^2(\Bbb R)$. In fact, let us first introduce a third set of functions of $\D$, $\F_\Phi=\{\Phi_n(x),\,n\geq0\}$, where
$$
\Phi_n(x)=e^{2i\alpha n x} \,e^{-\frac{x^2}{2}+x\left(\gamma+\frac{3}{2}\,i\alpha\right)}.
$$
It is possible to see that the set $\F_\Phi$ is complete in $\Lc^2(\Bbb R)$ if and only if $\F_\varphi$ is complete, and this is true if and only if $\F_\Psi$ is complete. To prove now that the set $\F_\Phi$ is complete in $\Lc^2(\Bbb R)$, we consider a square-integrable function $f(x)$ which is orthogonal to all the $\Phi_n(x)$: $I_n(\alpha):=\left<f,\Phi_n\right>=0$, for all $n\geq0$. Here we have written explicitly the dependence of these scalar products on $\alpha$, since $\Phi_n(x)$ explicitly depends on $\alpha$ and this will be useful in the following. We can rewrite $I_n(\alpha)$ as follows:
$$
I_n(\alpha)=\int_{\Bbb R}g(x)\,e^{i\alpha x\left(2n+\frac{3}{2}\right)}\,dx,
$$
where $g(x)=\overline{f(x)} e^{-\frac{x^2}{2}+x\gamma}$. Since $I_n(\alpha)=0$, we also have $\frac{d I_n(\alpha)}{d\alpha}=0$.  Then
$$
0=\frac{d}{d\alpha}\int_{\Bbb R}g(x)e^{i\alpha x\left(2n+\frac{3}{2}\right)}\,dx=\int_{\Bbb R}\frac{\partial}{\partial\alpha}g(x)e^{i\alpha x\left(2n+\frac{3}{2}\right)}\,dx=i\left(2n+\frac{3}{2}\right)\int_{\Bbb R}g(x)xe^{i\alpha x\left(2n+\frac{3}{2}\right)}\,dx
$$
for all $n\geq0$. This follows from the fact that $\frac{d}{d\alpha}$ and $\int_{\Bbb R}$ can be exchanged, see \cite{jones}, pg. 154. Iterating this procedure, we can further check that
$$
\int_{\Bbb R}g(x)x^le^{i\alpha x\left(2n+\frac{3}{2}\right)}\,dx=0,
$$
for all $l,n=0,1,2,\ldots$ and for all $\alpha$. Since the dependence on $\alpha$ is continuous, then, see again \cite{jones}, pg. 153, we also have $\int_{\Bbb R}g(x)x^l\,dx=0$ for all $l\geq0$, which can be written as
$$
\int_{\Bbb R}\left(\overline{f(x)} e^{-\frac{x^2}{4}+x\gamma}\right)\eta_l(x)\,dx=0
$$
for all $l=0,1,2,\ldots$. Here $\eta_l(x)=x^l\,e^{-\frac{x^2}{4}}$. But the set $\{\eta_l(x)\}$ is complete in $\Lc^2(\Bbb R)$, \cite{kolfom}\footnote{Our claim is based on the following result, given in \cite{kolfom}: if $\rho(x)$ is a Lebesgue-measurable function which is different from zero almost everywhere (a.e.) in $\Bbb R$ and if there exist two positive constants $\delta, C$ such that $|\rho(x)|\leq C\,e^{-\delta|x|}$ a.e. in $\Bbb R$, then the set $\left\{x^n\,\rho(x)\right\}$ is complete in $\Lc^2(\Bbb{R})$.}. Hence $\overline{f(x)} e^{-\frac{x^2}{4}+x\gamma}=0$ almost everywhere (a.e.) in $\Bbb R$, so that $f(x)=0$ a.e. as well. This is what we had to prove.

\vspace{2mm}

If we now compare the vacua of $a$ and $b^\dagger$, $\varphi_0(x)$, $\Psi_0(x)$, with the vacuum $\phi_0(x)=\frac{1}{\pi^{1/4}}e^{-\frac{x^2}{2}+\frac{3}{2}\,i\alpha x}$ of $c$,  \cite{eremel}, we easily see that
$$
\frac{\varphi_0(x)}{\phi_0(x)}=e^{\gamma x}, \qquad \mbox{while }\qquad \frac{\Psi_0(x)}{\phi_0(x)}=e^{-\gamma x}.
$$
This same ratios are preserved for larger values of $n$ so that we can write
\be
\varphi_n(x)=S\phi_n(x),\qquad \Psi_n(x)=S^{-1}\phi_n(x),
\label{55}\en
for all $n\geq0$, where $S$ is the multiplication operator $S=e^{\gamma x}$. Of course $S$ is invertible, unbounded with unbounded inverse, and leaves, together with $S^{-1}$, $\D$ stable\footnote{This is a consequence of the definition of $\D$: if $f(x)\in\D$, then $f(x)\in\Sc(\Bbb R)$ and $e^{kx}f(x)\in\Sc(\Bbb R)$ for all real $k$. Calling $f_{\pm1}(x)=S^{\pm 1}f(x)=e^{\pm kx}f(x)$, this is clearly in $\Sc(\Bbb R)$ as well as $e^{kx}f_{\pm1}(x)$. Hence $f_{\pm 1}(x)\in\D$.}. Also, $S=S^\dagger$. From (\ref{55}) the biorthogonality of the vectors in $\F_\varphi$ and $\F_\Psi$ immediately follows. Moreover, see \cite{baginbagbook}, (\ref{55}) also implies that $\F_\varphi$ and $\F_\Psi$ are $\D$-quasi bases, other than being complete in $\Lc^2(\Bbb R)$. What is not evident, is if they are also bases or not. This is because $S$ and $S^{-1}$ are unbounded. As for the operator $\Theta$ in (\ref{38}), this can be easily identified and it turns out that $\Theta=S^{-2}$, and $a$ and $b$ are $\Theta$-conjugate: $af=\Theta^{-1}b^\dagger\Theta f$, for all $f\in\D$. Using the closed form given in \cite{eremel} for $\phi_n(x)$ we can use (\ref{55}) to deduce the analytic expression for $\varphi_n(x)$ and $\Psi_n(x)$. In particular, in this way we recover the functions $\varphi_1(x)$ and $\varphi_2(x)$ deduced before.

\vspace{2mm}

It is now natural to look for the BCs associated to the operators $a$ and $b^\dagger$. It is natural to imagine that these BCs are related to the coherent state found in \cite{eremel} by the operators $S$ and $S^{-1}$. However, rather than exploring this relation, we prefer to use what deduced in Proposition \ref{thm1}, and to show that the existence of the vectors in (\ref{42}) can be explicitly proved. Using (\ref{55}), the expression of $\phi_n(x)$ given in \cite{eremel}, and the fact that $\Psi_n(x)$ can be recovered out of $\varphi_n(x)$ upon replacing $\gamma$ with $-\gamma$, simple computations show that
$$
\|\varphi_n\|^2=\|\Psi_n\|^2=\frac{[n]!e^{\gamma^2}}{(1-q)^n}L_n, \quad \mbox{where }\quad L_n=\sum_{k,l=0}^n (-1)^{k+l}\,\frac{e^{-\alpha^2(k+l+(l-k)^2)}e^{2i\alpha\gamma(l-k)}}{[k]![l]![n-k]![n-l]!}
$$
A very {\em brutal} estimate shows that $L_n\leq(n+1)^2$ so that
$$
\|\varphi_n\|\leq e^{\frac{\gamma^2}{2}}\,(n+1)\,\sqrt{\frac{[n]!}{(1-q)^n}}.
$$
The same bound holds for $\|\Psi_n\|$. Then the assumptions of Proposition \ref{thm1} on the norms of $\varphi_n$ and $\Psi_n$ are satisfied  identifying, for instance $A_\varphi=e^{\frac{\gamma^2}{2}}$, $r_\varphi=\frac{1}{\sqrt{1-q}}$ and $M_n(\varphi)=(n+1)\sqrt{[n]!}$. Indeed we have
$$
M(\varphi)=\lim_{n,\infty}\frac{M_n(\varphi)}{M_{n+1}(\varphi)}=\lim_{n,\infty}\frac{n+1}{\sqrt{[n+1]}\,(n+2)}=\sqrt{1-q}.
$$
Since $\sqrt{1-q}<\frac{1}{\sqrt{1-q}}$, it is clear that $\rho_\varphi=\sqrt{1-q}$. Analogously we find that $\rho_\Psi=\sqrt{1-q}$, so that $\rho=\sqrt{1-q}$ as well. Hence the BCs introduced in (\ref{42}) exist for $|z|<\rho$, and Proposition \ref{thm1} can be applied.

\subsection{A bounded deformation}\label{subsectex2}

Equations in (\ref{55}) show that the vectors $\varphi_n$ and $\Psi_n$ are similar to the $\phi_n$'s, but the similarity is implemented by an unbounded operator with unbounded inverse. This implies, as stated, that the sets $\F_\varphi$ and $\F_\Psi$ are not necessarily bases, even if they turn out to be complete and $\D$-quasi bases. We will discuss here a different example where the basis property is preserved. For this, we will use a stronger version of (\ref{55}), in which the similarity operator will be taken to be bounded with bounded inverse. This will produce, automatically, Riesz bases for the Hilbert space.

The starting point of this example is a pair of vectors $u$ and $v$ of $\Hil=l^2(\Bbb N_0)$ such that $\left<u,v\right>=1$. These vectors define a (non orthogonal) projection operator $P_{u,v}$ as follows: $P_{u,v}f=\left<u,f\right>\,v$. It is easy to check that $P_{u,v}^2=P_{u,v}$, and that $\|P_{u,v}\|=1$. Its adjoint acts on $f\in\Hil$ as follows: $P_{u,v}^\dagger f=\left<v,f\right>\,u$. Hence  $P_{u,v}^\dagger=P_{v,u}$ and, of course, $\|P_{u,v}^\dagger\|=1$ as well. Let us now consider two complex numbers, $\alpha$ and $\beta$, satisfying the equality $\alpha+\beta+\alpha\beta=0$. For instance, one could take $\alpha=i$ and $\beta=-\frac{i+1}{2}$. Then the operator $S=\1+\alpha P_{u,v}$ is bounded and invertible, with bounded inverse $S^{-1}=\1+\beta P_{u,v}$. Moreover we have $S^\dagger=\1+\overline{\alpha}\,P_{v,u}$ and $(S^\dagger)^{-1}=\1+\overline{\beta}\,P_{v,u}$ which are also both bounded, clearly. As in (\ref{55}) we can now use these operators to construct two biorthogonal sets of vectors of $\Hil$, acting on the o.n. basis $\F_e$, see Section \ref{subsetcI1} :
\be
\varphi_n=Se_n=e_n+\alpha\left<u,e_n\right>\,v, \qquad \Psi_n=(S^\dagger)^{-1}e_n=e_n+\overline{\beta}\left<v,e_n\right>\,u.
\label{56}\en
Since both $S$ and $S^{-1}$ are bounded, the sets $\F_\varphi=\{\varphi_n\}$ and $\F_\Psi=\{\Psi_n\}$ are Riesz bases. Now, formula (\ref{22}) can be rewritten as $S^{-1}\varphi_{n+1}=\frac{1}{\beta_n}\,c^\dagger S^{-1}\varphi_n$, for each $n\geq0$, which suggests to introduce a raising operator as $b=Sc^\dagger S^{-1}$. In fact, with this definition we can write $\varphi_{n+1}=\frac{1}{\beta_n}\,b\,\varphi_n$. We will consider the domain of $b$ (and of $a$, see below) later in this section. Now, recalling that $c\,e_0=0$, we also have $c\,S^{-1}\varphi_0=0$. This suggests to introduce a second operator, $a$, as $a=Sc S^{-1}$. In fact, in this way, we first have $a\varphi_0=0$ and, more important, $[a,b]_q=\1$, at least formally. We will make this statement rigorous later on. The operators $a$ and $b$ can now be written as follows:
\be
a=c+\alpha P_{c^\dagger u,v}+\beta P_{u,cv}+\alpha\beta P_{\left<cv,u\right>u,v},\quad
b=c^\dagger+\alpha P_{c u,v}+\beta P_{u,c^\dagger v}+\alpha\beta P_{\left<c^\dagger v,u\right>u,v}.
\label{57}\en
It might be worth noticing that, while $P_{u,v}$ and $P_{v,u}$ are projector operators, those in (\ref{57}) are not, in general. For instance $P_{u,cv}^2\neq P_{u,cv}$, at least if $\left<u,cv\right>\neq1$.

Now, let us call $\D$ the subset of all compactly supported sequences of $l^2(\Bbb N_0)$. This means that $f=(f_0,f_1,f_2,\ldots)\in l^2(\Bbb N_0)$ is in $\D$ only if  $f_j=0$ for all $j$ larger than a certain natural number\footnote{We  say that $f_j$ is zero {\em definitively}.}. Of course, since $\D$ contains the linear span of the $e_n$'s, which form an o.n. basis in $\Hil$, $\D$ is dense in $l^2(\Bbb N_0)$. For convenience from now on we take $u,v\in\D$. This means, see (\ref{56}), that $\varphi_n$ and $\Psi_n$ both belong to $\D$ as well. More generally, since for instance $Sf=f+\alpha\left<u,f\right>\,v$ for all $f\in\Hil$, then if, $f\in\D$,  $Sf\in\D$ as well. Hence $\D$ is stable under the action of $S$ and, for similar reasons, under the action of $S^{-1}$, $S^\dagger$ and $(S^{\dagger})^{-1}$.

It is now easy to check that $a$ and $b$, together with their adjoints, leave $\D$ stable and satisfy $[a,b]_qf=f$, for all $f\in\D$. Moreover, since $\varphi_0$ and $\Psi_0$ both belong to $\D$, it follows that Assumptions $\D$-pq 1. and 2. are both satisfied. Assumption $\D$-pq 3. is satisfied also in its stronger version, since $\F_\varphi$ and $\F_\Psi$ are, as already stated, biorthogonal Riesz bases. The operator $\Theta$ turns out to be $(SS^\dagger)^{-1}$. This is because $\Psi_n=\Theta\varphi_n$, and since $\varphi_n=Se_n$, $\Psi_n=(S^\dagger)^{-1}e_n$. We see that $\Theta$ is bounded, invertible and self-adjoint. Moreover, $\varphi_n$ is an eigenstate of $N=ba$ with eigenvalue $\beta^2_{n-1}$, while $\Psi_n$ is an eigenstate of $N^\dagger$ with the same eigenvalue, and $\Theta$ intertwines between the two.

As for the BCs, these can be easily constructed using Proposition \ref{thm1} since $\|\varphi_n\|\leq 1+|\alpha|$ and $\|\Psi_n\|\leq 1+|\beta|$, which are both uniform in $n$. More in general, since $\|\varphi_n\|\leq\|S\|$ and $\|\Psi_n\|\leq\|S^{-1}\|$, the estimates required by Proposition \ref{thm1} are easily satisfied: in fact, it is enough to put $A_\varphi=\|S\|$, $A_\Psi=\|S^{-1}\|$ and $r_\varphi=r_\Psi=M_n(\varphi)=M_n(\Psi)=1$.

 But, in this example, we can do better than this, due to the fact that $S$ and $S^{-1}$ are bounded. In fact, if we introduce the coherent states for the original quons $(c,c^\dagger)$,
$$
e(z)=N(|z|)\sum_{k=0}^\infty\frac{z^k}{\beta_{k-1}!}\,e_k,
$$
where $N(|z|)$ is defined in (\ref{41}), this is well defined if $|z|<\frac{1}{\sqrt{1-q}}$. Then we get $\varphi(z)=Se(z)$ and $\Psi(z)=(S^\dagger)^{-1}e(z)$. Hence they are Riesz-BCs in the sense of \cite{bagcs3}.

\subsubsection{An example in the example}

Let $I_0, I_1$ and $I_2$ be three finite subsets of $\Bbb N_0$, mutually disjoint, and let us define three vectors
$$
c_j=\sum_{k\in I_j}\gamma_k e_k,
$$
$j=1,2,3$, for some complex numbers $\gamma_k$. In particular we require that $\sum_{k\in I_0}|\gamma_k|^2=1$, which implies that $\|c_0\|=1$. Now, we use these vectors to define $u$ and $v$ and $P_{u,v}$ as before: we take $u=c_0+c_1$ and $v=c_0+c_2$. Then $\left<u,v\right>=1$ and we find that
$$
\varphi_k=\left\{
    \begin{array}{ll}
e_k, \qquad\qquad\quad\quad\,\, k\notin I_0\cup I_1,\\
e_k+\alpha\,\overline{\gamma_k}\,v\qquad\quad\,\, k\in I_0\cup I_1,\\
       \end{array}
        \right.
$$
while
$$
\Psi_k=\left\{
    \begin{array}{ll}
e_k, \qquad\qquad\quad\quad\,\, k\notin I_0\cup I_2,\\
e_k+\overline{\beta\,\gamma_k}\,u\qquad\quad\,\, k\in I_0\cup I_2.\\
       \end{array}
        \right.
$$
A direct computation shows that $\left<\varphi_k,\Psi_l\right>=\delta_{k,l}$, as it should. Moreover, the BCs can be written in the following closed form:
$$
\varphi(z)=e(z)+\alpha\,N(|z|)\Gamma_1(z)\,v,\qquad \Psi(z)=e(z)+\overline{\beta}\,N(|z|)\Gamma_2(z)\,u,
$$
where
$$
\Gamma_1(z)=\sum_{k\in I_0\cup I_1} \frac{z^k\,\overline{\gamma_k}}{\sqrt{\epsilon_k!}}, \qquad \Gamma_2(z)=\sum_{k\in I_0\cup I_2} \frac{z^k\,\overline{\gamma_k}}{\sqrt{\epsilon_k!}}.
$$
Needless to say, these two are finite sums so that there is no convergence issue to consider: we see that the Riesz BCs are, here, related to a single standard coherent state via two additive terms, proportional to the vectors $u$ and $v$ introduced at the beginning of this example.

\section{Conclusions}\label{sectconcl}

This paper continues the analysis of deformation of CCR, CAR, and other (anti)-commutation rules adopted in quantum mechanics during the years, and which proved to be essential in very many situations.

In particular, we have shown how the q-mutation relations can be deformed in a similar way as the one which produces pseudo-fermions, pseudo-bosons and truncated pseudo-bosons. This deformation opens interesting mathematical problems and suggests the use of $\D$-PQs in the context of pseudo-hermitian quantum mechanics (and its relatives). Also, we have shown that BCs states can be associated to $\D$-PQs, and that these states share with ordinary coherent states many of their properties.

Two examples have been discussed in details, one connected with an unbounded similarity map, and the other with a bounded one. In both cases we have been able to produce two biorthogonal sets of vectors, eigenstates of two number-like operators, but the nature of these sets depends on the similarity map appearing in our construction: completeness is not lost, in both cases, while the basis property can only be proved in the second example. However, in the first example, we still get $\D$-quasi bases which, as proved in \cite{baginbagbook}, have nice and useful properties.

\section*{Data accessibility statement}

This work does not have any experimental data.

\section*{Competing interests statement}

I have no competing interests.

\section*{Authors' contributions}

FB cured all the aspects of the paper.

\section*{Acknowledgements}

FB acknowledges  support by the GNFM of Indam. FB also gratefully thanks Prof. Eremin for useful discussions concerning Section \ref{subsectex1}.

\section*{Funding statement}

This work received no financial support.

\end{document}